\documentclass[12pt]{article}

\usepackage[T1]{fontenc}
\usepackage{lmodern}

\usepackage{setspace}
\usepackage[margin=1.25in]{geometry}

\usepackage[dvipsnames]{xcolor}
\usepackage{pdfpages}
\usepackage{float}
\usepackage{multirow}
\usepackage{caption}
\usepackage{subcaption}
\usepackage{epigraph}

\usepackage{mathtools}
\usepackage{amssymb}
\usepackage{amsthm}
\usepackage{amsfonts}
\usepackage{aliascnt}
\usepackage{accents}
\usepackage{dutchcal}
\usepackage{bbm}

\usepackage{enumitem}
\usepackage{sgame}
\usepackage{tikz}
\usepackage{tikz-cd}
\usetikzlibrary{calc,shapes,arrows}
\usepackage{tcolorbox}
\usepackage{listings}
\usepackage[normalem]{ulem}

\usepackage[round]{natbib}
\newcommand{\N}{\mathbb{N}}

\newtheorem{theorem}{Theorem}[section]

\newaliascnt{proposition}{theorem}
\newtheorem{proposition}[proposition]{Proposition}
\aliascntresetthe{proposition}

\newaliascnt{lemma}{theorem}
\newtheorem{lemma}[lemma]{Lemma}
\aliascntresetthe{lemma}

\newaliascnt{corollary}{theorem}
\newtheorem{corollary}[corollary]{Corollary}
\aliascntresetthe{corollary}

\newaliascnt{claim}{theorem}

\aliascntresetthe{claim}

\theoremstyle{definition}

\newaliascnt{definition}{theorem}
\newtheorem{definition}[definition]{Definition}
\aliascntresetthe{definition}

\newaliascnt{example}{theorem}

\aliascntresetthe{example}

\newaliascnt{assumption}{theorem}
\newtheorem{assumption}[assumption]{Assumption}
\aliascntresetthe{assumption}

\newaliascnt{condition}{theorem}

\aliascntresetthe{condition}

\newaliascnt{question}{theorem}

\aliascntresetthe{question}

\newaliascnt{remark}{theorem}

\aliascntresetthe{remark}

\newaliascnt{remarks}{theorem}

\aliascntresetthe{remarks}

\newaliascnt{aside}{theorem}

\aliascntresetthe{aside}

\newaliascnt{note}{theorem}

\aliascntresetthe{note}

\usepackage{thmtools}
\usepackage{thm-restate}

\usepackage{hyperref}

\hypersetup{
    colorlinks=true,
    linkcolor=OrangeRed,
    filecolor=Thistle,
    urlcolor=Thistle,
    citecolor=Thistle,
}

\usepackage[nameinlink]{cleveref}

\crefname{theorem}{theorem}{theorems}
\Crefname{theorem}{Theorem}{Theorems}
\crefname{proposition}{proposition}{propositions}
\Crefname{proposition}{Proposition}{Propositions}
\crefname{lemma}{lemma}{lemmas}
\Crefname{lemma}{Lemma}{Lemmas}
\crefname{corollary}{corollary}{corollaries}
\Crefname{corollary}{Corollary}{Corollaries}
\crefname{claim}{claim}{claims}
\Crefname{claim}{Claim}{Claims}
\crefname{definition}{definition}{definitions}
\Crefname{definition}{Definition}{Definitions}
\crefname{example}{example}{examples}
\Crefname{example}{Example}{Examples}
\crefname{assumption}{assumption}{assumptions}
\Crefname{assumption}{Assumption}{Assumptions}

\makeatletter
\let\cref@old@isrefconsecutive\cref@isrefconsecutive
\def\cref@isrefconsecutive#1#2{%
  \begingroup
    \def\cref@assumptiontype{assumption}%
    \cref@gettype{#1}{\cref@typea}%
    \ifx\cref@typea\cref@assumptiontype
      \endgroup
      \@cref@refconsecutivefalse
    \else
      \endgroup
      \cref@old@isrefconsecutive{#1}{#2}%
    \fi
}
\makeatother

\crefname{condition}{condition}{conditions}
\Crefname{condition}{Condition}{Conditions}
\crefname{question}{question}{questions}
\Crefname{question}{Question}{Questions}
\crefname{remark}{remark}{remarks}
\Crefname{remark}{Remark}{Remarks}
\crefname{remarks}{remarks}{remarks}
\Crefname{remarks}{Remarks}{Remarks}
\crefname{aside}{aside}{asides}
\Crefname{aside}{Aside}{Asides}
\crefname{note}{note}{notes}
\Crefname{note}{Note}{Notes}
\crefname{appendix}{appendix}{appendices}
\Crefname{appendix}{Appendix}{Appendices}

\newcommand{\secref}[1]{\hyperref[#1]{\S\ref*{#1}}}

\definecolor{backcolour}{rgb}{0.63,0.79,0.95}
\lstdefinestyle{mystyle}{
  backgroundcolor=\color{backcolour},
  basicstyle=\ttfamily\footnotesize,
  breakatwhitespace=false,
  breaklines=true,
  captionpos=b,
  keepspaces=true,
  numbers=left,
  numbersep=5pt,
  showspaces=false,
  showstringspaces=false,
  showtabs=false,
  tabsize=2
}
\begin{document} 
\title{TANSTAAFE\footnote{There Ain’t No Such Thing as a Free Equilibrium}}
\author{Mark Whitmeyer\thanks{Arizona State University. Email: \href{mailto:mark.whitmeyer@gmail.com}{mark.whitmeyer@gmail.com}. For KS. I used ChatGPT as one would an RA.}}
\date{\today}
\maketitle

\begin{abstract}
I argue that there is a sense in which universal equilibrium (defined loosely) existence in games is incompatible with eschewing strictly dominated strategies and a sense in which it isn't.
\end{abstract}

\section{Equilibrium and its Discontents}

All is well in finite games: they possess equilibria in mixed strategies. Not so in infinite games. One response is to enlarge the class of probabilities from countably additive to finitely additive ones. This restores equilibrium existence, but it creates a sharp distinction: a finitely additive probability can assign probability zero to every action in an infinite set while assigning probability one to the set itself.

For dominance, this distinction can be pathological. A finitely additive equilibrium may assign probability zero to every strictly dominated action and yet assign probability one to the set of strictly dominated actions. Each dominated action is then ignored in isolation, even though play is concentrated entirely on dominated actions. Is this pathology specific to a particular equilibrium concept?

No. I show that this distinction is unavoidable for dominated actions. No solution rule--even a set-valued one--can be universally nonempty, satisfy a minimal equilibrium coherence requirement for zero-sum games, and ignore the entire set of actions strictly dominated by finite lotteries.

A weaker requirement can always be met. Rather than asking that all dominated actions be ignored at once (which I term set-nullity), ask only that each be ignored separately (point-nullity). Every bounded game with finitely many players has a Nash equilibrium in finitely additive probabilities with this property. The result is sharp: in the counterexample I construct to thwart set-nullity, every dominated action is ignored, yet the dominated set itself receives probability one.

\section{Setup \& Preliminaries}\label{sec:evaluations-dominance}

A normal-form game, here, simply a \textit{game}, is a triple \(\Gamma=(I,(A_i)_{i\in I},(u_i)_{i\in I})\). As is standard, I write \(A\coloneqq\times_{i\in I}A_i\), \(A_{-i}\coloneqq\times_{j\neq i}A_j\), and \(a=(a_i,a_{-i})\in A\). I impose the following minimal standing assumption:
\begin{assumption}\label{ass:bounded-game}
The player set \(I\) is finite and nonempty, each action set \(A_i\) is nonempty, and each payoff function \(u_i\colon A\to\mathbb{R}\) is bounded.
\end{assumption}

For a nonempty set \(X\), let \(\ell^\infty(X)\) denote the Banach space of bounded real-valued functions on \(X\), equipped with the \(\sup\) norm \(\lVert f\rVert_\infty \coloneqq \sup_{x\in X}|f(x)|\). For \(B\subseteq X\), let \(\mathbf 1_B\in\ell^\infty(X)\) denote its indicator. 

\begin{definition}\label{def:evaluation-nullity}
An \textit{evaluation} on a nonempty set \(X\) is a map
\(\Lambda\colon\ell^\infty(X)\to\mathbb R\).
\end{definition}

A set \(B\subseteq X\) is \textit{\(\Lambda\)-null} if, for every
\(f,g\in\ell^\infty(X)\),
\[f=g\text{ on }X\setminus B \qquad\Longrightarrow\qquad \Lambda(f)=\Lambda(g).\] 
That is, altering consequences on \(B\) does not affect the evaluation. For nonempty \(X\) and an evaluation \(\Lambda\colon\ell^\infty(X)\to\mathbb R\), define
\[\mathcal N_\Lambda \coloneqq \{B\subseteq X\colon B\text{ is }\Lambda\text{-null}\}.\]

\begin{definition}
    An \textit{ideal}, \(\mathcal I\), on \(X \neq \emptyset\) is a set of subsets of \(X\) that is
\begin{enumerate}[noitemsep]
    \item Closed under finite unions: \(\emptyset \in \mathcal I\) and \(S_1, S_2 \in \mathcal I\) \(\Longrightarrow\) \(S_1 \cup S_2 \in \mathcal I\); and 
    \item Downward-closed: \(S_2 \in \mathcal I\) and \(S_1 \subseteq S_2\) \(\Longrightarrow\) \(S_1 \in \mathcal I\).
\end{enumerate}
An ideal is a \textit{proper ideal} if, moreover, \(X \notin \mathcal I\).
\end{definition}

\begin{lemma}\label{lem:null-ideal}
The set \(\mathcal N_\Lambda\) is an ideal. It is a proper ideal if and only if \(\Lambda\) is nonconstant.
\end{lemma}

\begin{proof}[Proof of \Cref{lem:null-ideal}]
If \(f=g\) on \(X\), \(\Lambda(f)=\Lambda(g)\). Hence, the empty set is \(\Lambda\)-null.

Suppose \(B\) is \(\Lambda\)-null and \(C\subseteq B\). If \(f=g\) on
\(X\setminus C\), then \(f=g\) on \(X\setminus B\), so
\(\Lambda(f)=\Lambda(g)\). Thus, \(C\) is \(\Lambda\)-null.

Now suppose \(B\) and \(C\) are \(\Lambda\)-null and \(f=g\) on \(X\setminus(B\cup C)\). Define \(h\in\ell^\infty(X)\) by
\[
h(x)
\coloneqq
\begin{cases}
g(x), \quad &\text{if} \quad x\in B,\\
f(x), \quad &\text{if} \quad x\notin B.
\end{cases}
\]
The functions \(f\) and \(h\) agree on \(X\setminus B\), so
\(\Lambda(f)=\Lambda(h)\). Moreover, \(h=g\) on \(B\setminus C\) by construction, and \(h=f=g\) on \(X\setminus(B\cup C)\). Hence, \(h\) and \(g\) agree on \(X\setminus C\), so \(\Lambda(h)=\Lambda(g)\). Therefore, \(\Lambda(f)=\Lambda(g)\), and \(B\cup C\) is \(\Lambda\)-null.

Finally, because \(X\setminus X=\emptyset\), any two functions agree on \(X\setminus X\). Consequently, \(X\) is \(\Lambda\)-null if and only if \(\Lambda(f)=\Lambda(g)\) for every \(f,g\in\ell^\infty(X)\), if and only if \(\Lambda\) is constant. 
\end{proof}

\begin{definition}\label{def:finite-lottery}
For \(X \neq \emptyset\), let \(\Delta_f(X)\) denote the set of \textit{finite lotteries}: finitely supported functions \(q\colon X\to[0,1]\) satisfying \(\sum_{x\in\operatorname{supp}(q)}q(x)=1\).
\end{definition}
For \(q_i\in\Delta_f(A_i)\) and
\(a_{-i}\in A_{-i}\), define
\[
u_i(q_i,a_{-i})
\coloneqq
\sum_{b_i\in\operatorname{supp}(q_i)}
q_i(b_i)u_i(b_i,a_{-i}).
\]

\begin{definition}\label{def:finite-strict-dominance}
An action \(a_i\in A_i\) is \textit{strictly dominated} in \(\Gamma\) if there is a finite lottery \(q_i\in\Delta_f(A_i)\) such that \(u_i(q_i,a_{-i})>u_i(a_i,a_{-i})\) for every \(a_{-i}\in A_{-i}\). Let \(D_i(\Gamma)\) denote the set of strictly dominated actions of player \(i\).
\end{definition}

For nonempty subsets \(B_i\subseteq A_i\), write \(\Gamma|_B\) for the restriction of \(\Gamma\) to \(B\coloneqq\times_{i\in I}B_i\).

\begin{definition}\label{def:iterated-finite-dominance-core}
Set \(A_i^0\coloneqq A_i\). Whenever \(A^k\coloneqq\times_{i\in I}A_i^k\) is nonempty, define \(A_i^{k+1} \coloneqq A_i^k\setminus D_i(\Gamma|_{A^k})\); and if \(A^k=\emptyset\), set \(A_i^{k+1}\coloneqq\emptyset\) for every \(i\in I\). The set of actions that survive every finite round is \(A_i^\infty \coloneqq \bigcap_{k=0}^{\infty}A_i^k\).
\end{definition}

\begin{definition}\label{def:point-set-nullity}
For each \(i\in I\), let \(\Lambda_i\colon\ell^\infty(A_i)\to\mathbb R\). The tuple \((\Lambda_i)_{i\in I}\) has \textit{point-nullity} for \(\Gamma\) if every singleton \(\{a_i\}\) with \(a_i\in D_i(\Gamma)\) is \(\Lambda_i\)-null. It has \textit{set-nullity} for \(\Gamma\) if \(D_i(\Gamma)\) is \(\Lambda_i\)-null for every \(i\in I\).
\end{definition}
\Cref{lem:null-ideal} reveals that set-nullity implies point-nullity. The converse can fail because null sets need not be closed under countable unions. On the other hand, they are equivalent when the sets of strictly-dominated actions are finite.

\begin{corollary}\label{cor:finite-point-set-nullity-equivalence}
Suppose \(D_i(\Gamma)\) is finite for every \(i\in I\). For every tuple \((\Lambda_i)_{i\in I}\), where \(\Lambda_i\colon\ell^\infty(A_i)\to\mathbb R\), point-nullity for \(\Gamma\) is equivalent to set-nullity for \(\Gamma\).
\end{corollary}

\begin{proof}[Proof of
\Cref{cor:finite-point-set-nullity-equivalence}]
Set-nullity implies point-nullity because nullity is downward closed. Conversely, suppose the tuple has point-nullity. For each \(i\in I\), \(D_i(\Gamma) = \bigcup_{a_i\in D_i(\Gamma)}\{a_i\}\), which is a finite union of \(\Lambda_i\)-null sets, so \Cref{lem:null-ideal} implies that \(D_i(\Gamma)\) is \(\Lambda_i\)-null. Hence, the tuple has set-nullity.
\end{proof}

\begin{definition}\label{def:equilibrium-solution-correspondence}
For a class of games (satisfying \Cref{ass:bounded-game}), \(\mathcal G\), a \textit{solution correspondence} \(S\) on \(\mathcal G\) assigns to each \(\Gamma \in \mathcal{G}\) a set \(S(\Gamma)\). Every element of \(S(\Gamma)\) is a tuple \((\Lambda_i)_{i\in I}\), where \(\Lambda_i\colon\ell^\infty(A_i)\to\mathbb R\) for every \(i\in I\). The correspondence 
\begin{itemize}[noitemsep]
    \item is \textit{universally nonempty} if \(S(\Gamma)\neq\emptyset\) for every \(\Gamma\in\mathcal G\); and
    \item satisfies \textit{point-nullity (set-nullity)} if, for every \(\Gamma\in\mathcal G\), every \((\Lambda_i)_{i\in I}\in S(\Gamma)\) has point-nullity (set-nullity).
\end{itemize}
\end{definition}
For a bounded two-player zero-sum game \(\Gamma=(\{1,2\},(A_1,A_2),(u_1,u_2))\), with \(u_2=-u_1\), let \(\mathcal E(\Gamma)\) denote the set of pairs \((\Lambda_1,\Lambda_2)\) (\(\Lambda_i\colon\ell^\infty(A_i)\to\mathbb R\) for \(i \in \left\{1,2\right\}\)) such that \(\Lambda_i(c\mathbf 1_{A_i})>0\) for every \(i \in \{1,2\}\) and every \(c > 0\) and
\[\Lambda_2(u_1(a_1,\cdot))
+
\Lambda_1(u_2(\cdot,a_2))
\leq0
\qquad
\text{for every }(a_1,a_2)\in A_1\times A_2.
\tag{Z}\label{eq:assessment-condition}
\]

Define \[\mathcal E^{\mathrm{set}}(\Gamma) \coloneqq
\left\{(\Lambda_1,\Lambda_2)\in\mathcal E(\Gamma) \colon (\Lambda_1,\Lambda_2) \text{ has set-nullity for } \Gamma\right\}.\]

\begin{definition}\label{def:zero-sum-equilibrium-condition} A solution correspondence \(S\) on \(\mathcal G\) has the \textit{equilibrium property} if \(S(\Gamma)\subseteq\mathcal E(\Gamma)\) for every bounded two-player zero-sum game \(\Gamma\in\mathcal G\). \end{definition}

Before going on, \eqref{eq:assessment-condition} deserves some explanation. In short, it is the aforementioned coherence requirement for zero-sum games. To wit, suppose, for the moment, that each \(\Lambda_i\) is the expectation functional induced by player \(i\)'s mixed strategy in a Nash equilibrium, and let \(v_i\) denote player \(i\)'s equilibrium payoff. For every \(a_1\in A_1\) and \(a_2\in A_2\), the absence of profitable unilateral deviations implies \(\Lambda_2\left(u_1(a_1,\cdot)\right)\leq v_1\) and \(\Lambda_1\left(u_2(\cdot,a_2)\right)\leq v_2\). And because the game is zero-sum, \(v_1+v_2=0\). Adding the two inequalities, therefore, produces
\[
\Lambda_2\left(u_1(a_1,\cdot)\right)
+
\Lambda_1\left(u_2(\cdot,a_2)\right)
\leq 0,
\]
which is precisely \eqref{eq:assessment-condition}. Thus, \eqref{eq:assessment-condition} retains a necessary consequence of zero-sum equilibrium while allowing the evaluations \(\Lambda_1\) and \(\Lambda_2\) to be more general than expectations under mixed strategies.\footnote{The requirement that positive constant functions receive strictly positive evaluations is a separate nondegeneracy condition. Moreover, \eqref{eq:assessment-condition} can be weakened: it is enough for my purposes to require, for every \((a_1,a_2)\in A_1\times A_2\), that the two quantities \(\Lambda_2\left(u_1(a_1,\cdot)\right)\) and \(\Lambda_1\left(u_2(\cdot,a_2)\right)\) not both be positive, which is (of course) necessary at any Nash equilibrium of a zero-sum game.}

\section{The Problem With Set-Nullity}\label{sec:setwise-obstruction}

I now construct a bounded two-player zero-sum game \(\Gamma^*\) such that \(\mathcal E^{\mathrm{set}}(\Gamma^*) = \emptyset\).

For every \(n \in \N\), set \(r_n\coloneqq n/(2(n+1))\) (and note that \(0<r_1<r_2<\cdots<1/2\)). Let \(C\coloneqq\{\mathrm H,\mathrm T\}\), and give each player the action set \(A_i^\ast\coloneqq C\cup \N\). Define
player \(1\)'s payoff by
\[
u^\ast(a_1,a_2)
\coloneqq
\begin{cases}
1,
\quad &\text{if} \quad a_1=a_2\in C,\\
-1,
\quad &\text{if} \quad a_1,a_2\in C\text{ and }a_1\neq a_2,\\
r_n,
\quad &\text{if} \quad a_1=n\in \N\text{ and }a_2\in C,\\
-r_m,
\quad &\text{if} \quad a_1\in C\text{ and }a_2=m\in \N,\\
r_n-r_m,
\quad &\text{if} \quad a_1=n\in \N\text{ and }a_2=m\in \N.
\end{cases}
\tag{LC}\label{eq:ladder-core-payoff}
\]
Player \(2\)'s payoff is \(u_2^\ast\coloneqq-u^\ast\). Denote the
resulting game by \(\Gamma^\ast\).

\begin{lemma}\label{lem:ladder-core-dominance}
For each player \(i\in\{1,2\}\), \(D_i(\Gamma^\ast) = \N\), and \(A_i^\infty = A_i^k = C\) for all \(k \geq 1\).\end{lemma}

\begin{proof}[Proof of \Cref{lem:ladder-core-dominance}]
For player \(1\), \eqref{eq:ladder-core-payoff} implies
\[
u^\ast(n+1,a_2)-u^\ast(n,a_2)=r_{n+1}-r_n>0
\]
for every \(n\in \N\) and every \(a_2\in A_2^\ast\); i.e., \(n+1\) strictly dominates \(n\). Analogously, for player \(2\) and every \(m \in \N\), \(m+1\) strictly dominates \(m\).

On the other hand, no action in \(C\) is strictly dominated. Player \(1\)'s action \(\mathrm H\) yields payoff \(1\) against \(\mathrm H\), and \(\mathrm T\) yields payoff \(1\) against \(\mathrm T\). Since \(1\)
is the largest payoff available to player \(1\), no finite lottery can be strictly better than either action at its respective witnessing
opponent action. The analogous reasoning works for player \(2\).

Consequently, \(D_i(\Gamma^\ast)=\N\), and the first deletion round leaves exactly \(C\). No action is strictly dominated in the restricted game on \(C\times C\). Consequently, \(A_i^k=C\) for every \(k\geq1\).
\end{proof}

\begin{proposition}\label{prop:ladder-core-obstruction}
We have \(\mathcal E^{\mathrm{set}}(\Gamma^\ast)=\emptyset\).
\end{proposition}

\begin{proof}[Proof of \Cref{prop:ladder-core-obstruction}]
Suppose for the sake of contradiction that \((\Lambda_1,\Lambda_2) \in \mathcal E^{\mathrm{set}}(\Gamma^\ast)\). By \Cref{lem:ladder-core-dominance}, \(\N\) is
\(\Lambda_i\)-null for both \(i\in\{1,2\}\).

The functions \(u^\ast(1,\cdot)\) and
\(r_1\mathbf 1_{A_2^\ast}\) agree on
\(C=A_2^\ast\setminus \N\). Since \(\N\) is \(\Lambda_2\)-null,
\[
\Lambda_2(u^\ast(1,\cdot))
=
\Lambda_2(r_1\mathbf 1_{A_2^\ast})
>
0.
\]

Likewise, the functions \(u_2^\ast(\cdot,1)\) and
\(r_1\mathbf 1_{A_1^\ast}\) agree on
\(C=A_1^\ast\setminus \N\). Since \(\N\) is \(\Lambda_1\)-null,
\[
\Lambda_1(u_2^\ast(\cdot,1))
=
\Lambda_1(r_1\mathbf 1_{A_1^\ast})
>
0.
\]

Consequently,
\[
\Lambda_2(u^\ast(1,\cdot))
+
\Lambda_1(u_2^\ast(\cdot,1))
>
0,
\]
contradicting \eqref{eq:assessment-condition} with
\((a_1,a_2)=(1,1)\). Therefore,
\(\mathcal E^{\mathrm{set}}(\Gamma^\ast)=\emptyset\).
\end{proof}

\begin{theorem}\label{thm:universal-setwise-impossibility} Let \(\mathcal G\) be any class of games, each satisfying \Cref{ass:bounded-game}, that contains \(\Gamma^\ast\). No solution correspondence on \(\mathcal G\) simultaneously i. has the equilibrium property, ii. is universally nonempty, and iii. satisfies set-nullity.
\end{theorem}

\begin{proof}[Proof of \Cref{thm:universal-setwise-impossibility}]
Suppose \(S\) has the equilibrium property, so that \(S(\Gamma^\ast)\subseteq\mathcal E(\Gamma^\ast)\); and satisfies set-nullity, so that every pair in \(S(\Gamma^\ast)\) belongs to \(\mathcal E^{\mathrm{set}}(\Gamma^\ast)\). Consequently, \(S(\Gamma^\ast) \subseteq \mathcal E^{\mathrm{set}}(\Gamma^\ast) = \emptyset\), by \Cref{prop:ladder-core-obstruction}, and \(S\) is not universally nonempty. \end{proof}

\section{Point-Nullity as a Panacea}
\label{sec:finite-approximations}

The previous section reveals that if we want universal equilibrium existence, we cannot always require the entire set of strictly dominated actions to be null. I now show that a weaker requirement is possible: each strictly dominated action can be null on its own.

To obtain this result, I consider Nash equilibria of finite restrictions of games and take limits as those restrictions include more and more actions. The resulting evaluations satisfy point-nullity. 

Continue to maintain \Cref{ass:bounded-game}.

\begin{definition}\label{def:finitely-additive-probability}
A \textit{finitely additive probability} on a nonempty set \(X\) is a linear map \(\Lambda\colon\ell^\infty(X)\to\mathbb R\) such that \(\Lambda(\mathbf 1_X)=1\) and \[f\geq0 \quad\Longrightarrow\quad \Lambda(f)\geq0.\]

Let \(\mathcal P^{\mathrm{fa}}(X)\) denote the set of 
finitely additive probabilities on \(X\), endowed with the weak-* topology \(\sigma(\ell^\infty(X)^*,\ell^\infty(X))\).
\end{definition}

Note the standard fact that for \(X \neq \emptyset\), \(\mathcal P^{\mathrm{fa}}(X)\) is weak-* compact.\footnote{Every \(\Lambda\in\mathcal P^{\mathrm{fa}}(X)\) satisfies \(|\Lambda(f)|\leq\lVert f\rVert_\infty\) for every \(f\in\ell^\infty(X)\) and, thus, belongs to \(\ell^\infty(X)^*\) with operator norm one. Then appeal to the Banach-Alaoglu theorem to conclude that \(\mathcal P^{\mathrm{fa}}(X)\) is weak-* compact.} For \(\widehat\Lambda\in\mathcal P^{\mathrm{fa}}(A)\), define its
\(i\)-th marginal
\(\widehat\Lambda_i\in\mathcal P^{\mathrm{fa}}(A_i)\) by
\[
\widehat\Lambda_i(f)\coloneqq\widehat\Lambda(a\mapsto f(a_i))
\qquad
\text{for every }f\in\ell^\infty(A_i).
\]
Importantly, weak-* convergence passes to marginals: whenever \(\widehat\Lambda^\alpha\to\widehat\Lambda\) weak-*, \(\widehat\Lambda_i^\alpha\to\widehat\Lambda_i\) weak-* for every \(i\in I\).

\begin{lemma}\label{lem:linear-nullity}
Let \(\Lambda\in\mathcal P^{\mathrm{fa}}(X)\). A set \(B\subseteq X\)
is \(\Lambda\)-null if and only if
\(\Lambda(\mathbf 1_B)=0\).
\end{lemma}

\begin{proof}[Proof of \Cref{lem:linear-nullity}]
If \(B\) is \(\Lambda\)-null, then \(\mathbf 1_B\) and the zero
function agree on \(X\setminus B\), so
\(\Lambda(\mathbf 1_B)=0\).

Conversely, suppose \(\Lambda(\mathbf 1_B)=0\), and let
\(f,g\in\ell^\infty(X)\) agree on \(X\setminus B\). Set
\(h\coloneqq f-g\). Since \(h\) vanishes on \(X\setminus B\),
\[
-\lVert h\rVert_\infty\mathbf 1_B\leq h\leq\lVert h\rVert_\infty\mathbf 1_B.
\]
Since \(\Lambda\) is positive and linear, applying \(\Lambda\) to these inequalities produces
\[
-\lVert h\rVert_\infty\Lambda(\mathbf 1_B)\leq\Lambda(h)\leq\lVert h\rVert_\infty\Lambda(\mathbf 1_B).
\]
Hence, \(\Lambda(h)=0\), so \(\Lambda(f)=\Lambda(g)\). Thus, \(B\) is
\(\Lambda\)-null.
\end{proof}
A \textit{directed set} is a nonempty set \(D\) with a reflexive and
transitive relation \(\succeq\) such that, for every
\(\alpha,\beta\in D\), there is \(\gamma\in D\) satisfying
\(\gamma\succeq\alpha\) and \(\gamma\succeq\beta\). A \textit{net}
is a family \((x^\alpha)_{\alpha\in D}\) indexed by a directed set. A
property holds \textit{eventually} if there is \(\alpha_0\in D\) such
that it holds for every \(\alpha\succeq\alpha_0\). A \textit{subnet} of \((x^\alpha)_{\alpha\in D}\) is a net
\((x^{\phi(\beta)})_{\beta\in E}\), where \(E\) is directed and
\(\phi\colon E\to D\) satisfies both of the following conditions:
\(\beta'\succeq\beta\) implies
\(\phi(\beta')\succeq\phi(\beta)\), and, for every
\(\alpha_0\in D\), \(\phi(\beta)\succeq\alpha_0\) eventually. A net \((\Lambda^\alpha)_{\alpha\in D}\) in
\(\mathcal P^{\mathrm{fa}}(X)\) converges weak-* to
\(\Lambda\in\mathcal P^{\mathrm{fa}}(X)\) if \(\Lambda^\alpha(f)\longrightarrow\Lambda(f)\) for every \(f\in\ell^\infty(X)\).

\begin{definition}\label{def:exhaustive-finite-restrictions} A net \((F^\alpha)_{\alpha\in D}\), where \(F^\alpha\coloneqq\times_{i\in I}F_i^\alpha\) and each \(F_i^\alpha\subseteq A_i\) is finite and nonempty, is \textit{exhaustive} for \(\Gamma\) if, for every tuple \((G_i)_{i\in I}\) of finite sets \(G_i\subseteq A_i\), there is \(\alpha_0\in D\) such that \(G_i\subseteq F_i^\alpha\) for every \(i \in I\) and every \(\alpha\succeq\alpha_0\). \end{definition}

\citet[Definitions~3.3-3.4 and p.~25]{KhanPedersenStinchcombe2026} define \(\operatorname{Eq}^{\mathrm{Fin}}(\Gamma)\) as the set of limits of equilibria of finite restrictions along exhaustive nets. Their convergence is the weak-* convergence used here \citep[\S3.2.2]{KhanPedersenStinchcombe2026}.

\begin{definition}\label{def:finitely-approximable-equilibrium}
A finitely additive probability \(\widehat\Lambda\in\mathcal P^{\mathrm{fa}}(A)\) is a \textit{finitely approximable equilibrium} for \(\Gamma\) if there are an exhaustive net \((F^\alpha)_{\alpha\in D}\) and, for every \(\alpha\in D\), a mixed Nash equilibrium \(\sigma^\alpha=(\sigma_i^\alpha)_{i\in I}\) of \(\Gamma|_{F^\alpha}\) such that the finitely additive probabilities
\[
\widehat\Lambda^\alpha(h)\coloneqq\sum_{a\in F^\alpha}
\left(\prod_{i\in I}\sigma_i^\alpha(a_i)
\right)h(a)
\qquad
\text{for every }h\in\ell^\infty(A),
\]
converge weak-* to \(\widehat\Lambda\).\footnote{Recall that for nonempty subsets \(B_i\subseteq A_i\), \(\Gamma|_B\) denotes the restriction of \(\Gamma\) to \(B\coloneqq\times_{i\in I}B_i\)} Let
\(\operatorname{Eq}^{\mathrm{Fin}}(\Gamma)\) denote the set of
finitely approximable equilibria for \(\Gamma\).
\end{definition}
Informally, a finitely approximable equilibrium is a weak-* limit of ordinary mixed Nash equilibria of finite restrictions that eventually contain every prescribed finite collection of actions.

Because \(I\) is finite, \Cref{ass:bounded-game} proffers a \(B>0\) such that \(\lVert u_i\rVert_\infty\leq B\) for every \(i\in I\). \citet[Theorem C]{KhanPedersenStinchcombe2026} states that \(\operatorname{Eq}^{\mathrm{Fin}}\) is nonempty-valued; i.e., \(\operatorname{Eq}^{\mathrm{Fin}}(\Gamma)\neq\emptyset\).

\begin{lemma}\label{lem:finitely-approximable-properties} If \(\widehat\Lambda\in\operatorname{Eq}^{\mathrm{Fin}}(\Gamma)\), then \(\widehat\Lambda(u_i)\geq\widehat\Lambda(a\mapsto u_i(b_i,a_{-i}))\) for every \(i \in I\) and every \(b_i\in A_i\). Moreover, the tuple \((\widehat\Lambda_i)_{i\in I}\) has point-nullity for \(\Gamma\).\end{lemma}

\begin{proof}[Proof of \Cref{lem:finitely-approximable-properties}]
Because \(\widehat\Lambda\in\operatorname{Eq}^{\mathrm{Fin}}(\Gamma)\), \Cref{def:finitely-approximable-equilibrium} provides a directed set \(D\), an exhaustive net \((F^\alpha)_{\alpha\in D}\), and, for every \(\alpha\in D\), a mixed Nash equilibrium \(\sigma^\alpha=(\sigma_i^\alpha)_{i\in I}\) of
\(\Gamma|_{F^\alpha}\), with \(\widehat\Lambda^\alpha\to\widehat\Lambda\) weak-*. Let
\(\widehat\Lambda_i^\alpha\) denote the \(i\)-th marginal of
\(\widehat\Lambda^\alpha\).

Fix \(i\in I\) and \(b_i\in A_i\). As \((F^\alpha)_{\alpha\in D}\) is an exhaustive net, \(b_i\in F_i^\alpha\) eventually. For every such \(\alpha\), the Nash equilibrium condition in \(\Gamma|_{F^\alpha}\) implies \(\widehat\Lambda^\alpha(u_i)\geq\widehat\Lambda^\alpha(a\mapsto u_i(b_i,a_{-i}))\), and weak-* convergence of \(\widehat\Lambda^\alpha\) to \(\widehat \Lambda\) yields
\(\widehat\Lambda(u_i)\geq\widehat\Lambda(a\mapsto u_i(b_i,a_{-i}))\).

Now fix \(a_i\in D_i(\Gamma)\). By
\Cref{def:finite-strict-dominance}, there is
\(q_i\in\Delta_f(A_i)\) that strictly dominates \(a_i\).
By exhaustiveness of \((F^\alpha)_{\alpha\in D}\), \(\{a_i\}\cup\operatorname{supp}(q_i)\subseteq F_i^\alpha\) eventually.

For every such \(\alpha\), the lottery \(q_i\) gives player \(i\) a strictly higher expected payoff than \(a_i\) against the other players' equilibrium lotteries. Hence, \(a_i\) is not a best response, and so \(\sigma_i^\alpha(a_i)=0\). Therefore, eventually, \(\widehat\Lambda_i^\alpha(\mathbf 1_{\{a_i\}})=\sigma_i^\alpha(a_i)=0\). Consequently,
\(\widehat\Lambda_i(\mathbf 1_{\{a_i\}})=0\), since
\(\widehat\Lambda_i^\alpha\to\widehat\Lambda_i\) weak-*.

By \Cref{lem:linear-nullity}, the singleton \(\{a_i\}\) is
\(\widehat\Lambda_i\)-null. Since \(i\) and \(a_i\) were arbitrary,
the tuple \((\widehat\Lambda_i)_{i\in I}\) has point-nullity.
\end{proof}

Let \(\mathcal G_{\mathrm{bd}}\) denote the class of all games
satisfying \Cref{ass:bounded-game}. \Cref{lem:finitely-approximable-properties} supplies the two properties needed for a positive result: the limiting joint assessment satisfies the pure-strategy-deviation inequalities, and its marginals make every strictly dominated action null. I, therefore, define the solution correspondence by retaining the marginals of finitely approximable equilibria.

\begin{definition}\label{def:point-null-solution-correspondence}For \(\Gamma\in\mathcal G_{\mathrm{bd}}\), define \(S^{\mathrm{pt}}(\Gamma)\coloneqq\left\{(\widehat\Lambda_i)_{i\in I} \colon \widehat\Lambda\in\operatorname{Eq}^{\mathrm{Fin}}(\Gamma) \right\}\).\end{definition}

\begin{theorem}\label{thm:universal-point-nullity}
On \(\mathcal G_{\mathrm{bd}}\), the correspondence
\(S^{\mathrm{pt}}\) is universally nonempty, satisfies point-nullity,
and has the equilibrium property.
\end{theorem}

\begin{proof}[Proof of \Cref{thm:universal-point-nullity}]
Universal nonemptiness follows from \(\operatorname{Eq}^{\mathrm{Fin}}\) being nonempty-valued \citep[Theorem~C]{KhanPedersenStinchcombe2026}.\footnote{Alternatively, here is a self-contained proof. Let \(D\) be the set of all \(F=\times_{i\in I}F_i\), where each \(F_i\subseteq A_i\) is finite and nonempty, and define \(F'\succeq F\) if \(F_i'\supseteq F_i\) for every \(i\in I\). The set \(D\) is directed, and the net \((F)_{F\in D}\) is exhaustive. Since each \(\Gamma|_F\) is finite, choose a mixed Nash equilibrium \(\sigma^F\) of \(\Gamma|_F\), and define \(\widehat\Lambda^F\) from \(\sigma^F\) as in \Cref{def:finitely-approximable-equilibrium}. Weak-* compactness of \(\mathcal P^{\mathrm{fa}}(A)\) means that the net \((\widehat\Lambda^F)_{F\in D}\) has a convergent subnet. By the definition of a subnet, the corresponding subnet of finite restrictions remains exhaustive. Its limit, therefore, belongs to \(\operatorname{Eq}^{\mathrm{Fin}}(\Gamma)\), so \(\operatorname{Eq}^{\mathrm{Fin}}(\Gamma)\neq\emptyset\).} Point-nullity follows from \Cref{lem:finitely-approximable-properties}.

It remains to verify the equilibrium property. Let \(\Gamma\) be a
bounded two-player zero-sum game with player set \(\{1,2\}\), action sets \(A_1,A_2\), and payoffs \(u_2=-u_1\). Take \((\widehat\Lambda_1,\widehat\Lambda_2) \in S^{\mathrm{pt}}(\Gamma)\). By \Cref{def:point-null-solution-correspondence}, there is \(\widehat\Lambda\in\operatorname{Eq}^{\mathrm{Fin}}(\Gamma)\) whose first and second marginals are \(\widehat\Lambda_1\) and \(\widehat\Lambda_2\), respectively. For \(i\in\{1,2\}\) and \(c>0\), as \(\widehat \Lambda_i\) is a finitely-additive probability, \(\widehat\Lambda_i(c\mathbf 1_{A_i})=c>0\).

Finally, fix \((b_1,b_2)\in A_1\times A_2\). By \Cref{lem:finitely-approximable-properties},
\[\widehat\Lambda_2(u_1(b_1,\cdot)) \leq \widehat\Lambda(u_1) \qquad \text{and} \qquad \widehat\Lambda_1(u_2(\cdot,b_2)) \leq\widehat\Lambda(u_2).\]
Consequently,
\[\widehat\Lambda_2(u_1(b_1,\cdot))+\widehat\Lambda_1(u_2(\cdot,b_2)) \leq \widehat\Lambda(u_1)+\widehat\Lambda(u_2) = \widehat\Lambda(u_1+u_2) = 0.\]
Thus, \((\widehat\Lambda_1,\widehat\Lambda_2)\in\mathcal E(\Gamma)\), and so \(S^{\mathrm{pt}}\) has the equilibrium property.
\end{proof}

\Cref{thm:universal-point-nullity} guarantees point-nullity, but leaves open whether finite approximability might also imply set-nullity. I now show that it does not.

\begin{corollary}\label{cor:ladder-tail-limit}
There exists \(\widehat\Lambda^\ast \in\operatorname{Eq}^{\mathrm{Fin}}(\Gamma^\ast)\) such that, for every \(i\in\{1,2\}\), \(\widehat\Lambda_i^\ast(\mathbf 1_{\{n\}})=0\) for every \(n\in\N\), while \(\widehat\Lambda_i^\ast(\mathbf 1_\N)=1\). Consequently, \((\widehat\Lambda_1^\ast,\widehat\Lambda_2^\ast) \in S^{\mathrm{pt}}(\Gamma^\ast)\) has point-nullity but not set-nullity.\end{corollary}

\begin{proof}[Proof of \Cref{cor:ladder-tail-limit}]
For \(k\in\N\), set \(F_i^k\coloneqq C\cup\{1,\ldots,k\}\) for \(i\in\{1,2\}\), and let \(F^k\coloneqq F_1^k\times F_2^k\). Every finite subset of \(A_i^\ast\) is contained in \(F_i^k\) for all sufficiently large \(k\), so \((F^k)_{k\in\N}\) is exhaustive for \(\Gamma^\ast\).

For \(n,m\in\{1,\ldots,k\}\) and \(c\in C\),
\[u^\ast(n,k) = r_n-r_k\leq0=u^\ast(k,k), \qquad u^\ast(c,k) =-r_k<0=u^\ast(k,k),\]
\[u_2^\ast(k,m) =r_m-r_k\leq0=u_2^\ast(k,k), \qquad \text{and} \qquad u_2^\ast(k,c) =-r_k<0=u_2^\ast(k,k),\] 
so \((k,k)\) is a Nash equilibrium of \(\Gamma^\ast|_{F^k}\).

Define \(\widehat\Lambda^k \in\mathcal P^{\mathrm{fa}}(A_1^\ast\times A_2^\ast)\) by \(\widehat\Lambda^k(h)\coloneqq h(k,k)\) for every \(h\in\ell^\infty(A_1^\ast\times A_2^\ast)\), and let \(\widehat\Lambda_i^k\) denote the \(i\)-th marginal of \(\widehat\Lambda^k\). Regard \((\widehat\Lambda^k)_{k\in\N}\) as a net indexed by \(\N\) under \(\geq\). Because every net in a compact space has a convergent subnet, there is a directed set \(E\), a subnet \((\widehat\Lambda^{k_\beta})_{\beta\in E}\), and a finitely additive probability \(\widehat\Lambda^\ast \in\mathcal P^{\mathrm{fa}}(A_1^\ast\times A_2^\ast)\) such that \(\widehat\Lambda^{k_\beta}\to\widehat\Lambda^\ast\) weak-*. By the definition of a subnet, for every \(K\in\N\), \(k_\beta\geq K\) eventually.

The net \((F^{k_\beta})_{\beta\in E}\) is, therefore, exhaustive. For every \(\beta\in E\), the finitely additive probability \(\widehat\Lambda^{k_\beta}\) is induced by the degenerate mixed Nash equilibrium concentrated on \((k_\beta,k_\beta)\), so \(\widehat\Lambda^\ast\in\operatorname{Eq}^{\mathrm{Fin}}(\Gamma^\ast)\).

For every \(k\in\N\) and \(i\in\{1,2\}\), \(\widehat\Lambda_i^k(\mathbf 1_\N)=1\). For each fixed \(n\in\N\), \(\widehat\Lambda_i^k(\mathbf 1_{\{n\}})=0\) whenever \(k>n\). Since \(\widehat\Lambda_i^{k_\beta}\to\widehat\Lambda_i^\ast\) weak-* and \(k_\beta>n\) eventually for each fixed \(n\in\N\), \(\widehat\Lambda_i^\ast(\mathbf 1_\N)=1\) and \(\widehat\Lambda_i^\ast(\mathbf 1_{\{n\}})=0\) for every \(n \in \N\). By \Cref{lem:linear-nullity,lem:ladder-core-dominance}, every dominated singleton is \(\widehat\Lambda_i^\ast\)-null, whereas the dominated set \(\N\) is not \(\widehat\Lambda_i^\ast\)-null.
\end{proof}

\begin{corollary}\label{cor:point-set-nullity-boundary}
The correspondence \(S^{\mathrm{pt}}\) is universally nonempty, has the equilibrium property, and satisfies point-nullity on \(\mathcal G_{\mathrm{bd}}\). If \(\mathcal G\) is a class of games satisfying \Cref{ass:bounded-game} that contains \(\Gamma^\ast\), no solution correspondence on \(\mathcal G\) has the equilibrium property while being universally nonempty and satisfying set-nullity. Moreover, \(S^{\mathrm{pt}}(\Gamma^\ast)\) contains a pair that assigns zero mass to every dominated singleton and mass one to each player's dominated set.
\end{corollary}

\begin{proof}[Proof of \Cref{cor:point-set-nullity-boundary}]
The first claim is \Cref{thm:universal-point-nullity}. The impossibility claim is \Cref{thm:universal-setwise-impossibility}. The final claim follows from \Cref{cor:ladder-tail-limit}.
\end{proof}

\section{Discussion}

My main result is motivated by a tension between two natural goals: finding an equilibrium in every bounded game and ensuring that dominated actions do not affect the solution. One influential route to the first goal enlarges the space of mixed strategies from countably additive to finitely additive probabilities. This ensures a compact strategy space even when action sets are arbitrary, but it also makes the payoff from mixed play potentially ambiguous: averaging first over one player's action and then over the other's need not provide the same answer as averaging in the opposite order.

Earlier papers deal with this problem in several ways.\footnote{\citet{Fenstad1967,Marinacci1997} identify conditions under which the two orders agree; \citet{Young1971,Thomsen1978} provide closely related criteria; and \citet{HarrisStinchcombeZame2005} seek representations by compact games with continuous payoffs.} Other papers keep a broader class of payoffs but change how they evaluate mixed play.\footnote{\citet{Yanovskaya1970,SchervishSeidenfeld1996} specify how payoffs are evaluated when the two orders of averaging disagree; \citet{FleschVermeulenZseleva2017,Vasquez2017} compare ``upper'' and ``lower'' possible payoffs; \citet{Milchtaich2023} requires strategies to concentrate on actions that are nearly best replies; and \citet{CapraroScarsini2013} obtain existence by imposing an algebraic structure.} A separate strand begins with finite games.\footnote{\citet{SimonStinchcombe1995} compare trembling-hand refinements with refinements obtained as limits of finite games; \citet{BajooriFleschVermeulen2013} strengthen the latter approach for compact action spaces; \citet{Stinchcombe2005} develops equilibrium through finite approximations and generalized integration; and \citet{KhanPedersenStinchcombe2026} establish universal existence via finitely-additive probabilities.} Special mention is due to \citet{FleschVermeulenZseleva2021}, who provide a close antecedent: in a modified \citet{Wald1945} game, their finitely additive equilibrium concept admits an equilibrium in which one player assigns probability zero to each member of a countable set of uniformly dominated actions and probability one to the set as a whole.\footnote{In the terminology of my paper, that equilibrium has point-nullity but not set-nullity.} But that is one particular solution concept--I ask whether there is some universally nonempty solution concept with the equilibrium property that has set-nullity. \Cref{prop:ladder-core-obstruction} and \Cref{thm:universal-setwise-impossibility} show that it cannot, whereas \Cref{thm:universal-point-nullity} shows that point-nullity is compatible with universal nonemptiness and the equilibrium property.

In short, my contribution is a sharp delineation rather than another equilibrium construction: finite approximation can eliminate each fixed dominated action, but no universally existing equilibrium solution can in general be required to ignore the entire, possibly infinite, set of dominated actions.

\bibliography{sample}

\end{document}